\documentclass[11pt]{article}

\usepackage[a4paper,margin=0.9in]{geometry}
\usepackage{amsmath,amssymb,amsthm}
\usepackage{natbib}
\defcitealias{jpart2026}{PMRA, 2026}
\usepackage[hidelinks]{hyperref}
\usepackage{microtype}
\usepackage{booktabs,tabularx}
\usepackage{docmute}

\newtheorem{lemma}{Lemma}
\newtheorem{proposition}{Proposition}
\newtheorem{corollary}{Corollary}

\newcommand{\E}{\mathbb{E}}
\newcommand{\Var}{\operatorname{Var}}
\newcommand{\Cov}{\operatorname{Cov}}
\newcommand{\CV}{\operatorname{CV}}
\newcommand{\Beta}{\operatorname{Beta}}
\newcommand{\1}{\mathbf{1}}

\newtheorem*{lemmaone}{Lemma 1 (Curation)}
\newtheorem*{propositionone}{Proposition 1 (Quotas versus equal-load fees)}
\newtheorem*{corollaryone}{Corollary 1 (Abundance can turn fees into an access screen)}

\newif\ifwithappendix
\withappendixtrue

\title{Submission quotas as self-screening\\ when manuscript production is cheap}
\author{Steve Lawford\thanks{ENAC (University of Toulouse), 7 avenue Edouard Belin, BP 54005, 31055 Toulouse Cedex 4, France (email: steve.lawford@enac.fr).}}

\date{}

\begin{document}

\maketitle

\begin{abstract}
\noindent Generative AI enables authors to produce more manuscripts at lower cost. But evaluation by expert editors and referees remains scarce. In this context, the value of an author's submission rights can increase. Major research repositories and journals are reacting by imposing submission quotas and fees. Non-transferable author quotas transform submission of a manuscript into a revealed-priority decision. Authors construct a portfolio of their highest-ranked manuscripts. A common fee charges an author for every manuscript they submit. Both methods have drawbacks. Quotas do not flexibly allocate scarce review capacity to authors who have unusually strong realised candidate sets. Fees can generate heterogeneous submission thresholds that may be systematically related to an author's baseline portfolio quality. In this paper, I derive the condition under which submission quotas yield weakly higher mean reviewed value than fees, given an equal expected review load. When project values are bounded and candidate manuscripts are abundant, an equal-load fee can price out some authors entirely.
\end{abstract}


\section{Introduction}

Generative AI can sharply reduce an author's cost of producing a manuscript. Expert editorial and referee capacity has not expanded comparably. This creates a technological asymmetry in the peer review process. Submissions to \emph{Organization Science} rose by 42\% following the launch of OpenAI's ChatGPT in November 2022. The increase was largely concentrated in manuscripts that contained substantial AI-generated text. Those submissions were much more likely to be desk rejected \citep{gartenberg2026}. Multiple studies document rapid growth in LLM-assisted and LLM-modified scientific writing \citep{liang2025,kusumegi2025}. Pressure on the review process is not new. Submissions to leading economics journals rose substantially between 1990 and 2012, while the number of published articles fell from the late 1970s onwards \citep{card2013}. An increase in the number of NBER working papers released in the same week reduces media attention for each paper and harms subsequent publication and citation outcomes \citep{lusher2023}.

Research platforms and journals are reacting by rationing authors' access to expert evaluation and moderation. On 1 October 2026, the widely used preprint repository arXiv imposed a submitter limit of two submissions per calendar month across all categories, in addition to its existing limit of three active submissions. Submissions to arXiv nearly doubled in the two years to September 2026, reaching 40,363 in that month; arXiv explicitly hopes that authors will ``curate and submit only their best work'' \citep{boboris2026}. The open-access electronic journal \emph{Transactions on Machine Learning Research} (TMLR) introduced an annual author-level generalised harmonic submission quota effective 1 July 2026, after its submissions tripled in a year; ``many individuals [submitted] four to six papers within one or two weeks'' \citep{tmlr2026}. The \emph{Journal of Public Administration Research and Theory} (JPART) began an author-level submission-cap pilot on 1 July 2026; from 1 January 2027, all listed authors will be limited to three submissions in a calendar year, with no more than two desk rejections \citepalias{jpart2026}. \emph{The Accounting Review} now limits individual authors to ``eight first-round submissions over a rolling 24-month period'' \citep{tar2026}. The \emph{International Conference on Learning Representations} (ICLR) has announced author-level submission limits for 2027 \citep{iclr2026}. Prices are another response: the journal \emph{California Management Review} introduced submission fees on 1 September 2026 ``in response to a rapidly changing scholarly-publishing environment and a sharp increase in manuscript volume'', noting that its existing submission caps had proved insufficient \citep{cmr2026}.

In this paper, I ask what a quota achieves beyond reducing the volume of submissions. If authors possess private information about the relative quality of their own manuscripts, then a non-transferable quota creates an opportunity cost of nominating one manuscript over another, without the platform needing to observe manuscript quality, detect AI use, or charge a submission fee. The quota I study uses ordinal information: it is robust to cardinal miscalibration but not to ranking errors. In a survey of submitting authors at a major computer science conference, authors substantially overestimated the acceptance probability of their papers (with a median prediction of 70\% for an approximately 25\% actual acceptance rate) \citep{rastogi2024}. By contrast, within-author rankings at the \emph{International Conference on Machine Learning} (ICML) outperformed peer-review scores in predicting future citation counts \citep{su2026impact}.

The paper is organised as follows. First, I show that expanding an author's candidate set cannot worsen their portfolio of highest-ranked manuscripts. Second, in a benchmark with uniformly distributed within-author project values, I compare quotas and fees at equal expected review load and derive a necessary and sufficient condition under which a quota yields higher mean reviewed value than a submission fee. I allow baseline portfolio quality, fee sensitivity and cardinal calibration to vary across authors. Third, given bounded project values, I show how, when candidate manuscripts become more abundant, an equal-load submission fee can price out some authors entirely. Although generative AI motivates this paper, it is not needed for the mechanism. Additional derivations, numerical examples, and boundary cases are available in the Supplementary Appendix.

\section{Related work}

It is well known in economics that submission costs can function as screening devices. \citet{azar2005} and \citet{leslie2005} observe that editorial delays impose a significant cost on authors that may discourage submission of low-quality manuscripts; \citet{cotton2013} allows editors to choose fees and delays jointly and finds that a mixture of moderate fees and delays provides optimal screening within his model. \citet{ellison2002} documents publication delays in economics; \citet{hadavand2024} report lengthy delays, with the mean time from manuscript submission to acceptance exceeding two years at some leading journals. \citet{fourie2026} concludes that submission fees reduce system-wide demand for reviewers only when some authors are persuaded to shorten their submission chains; \citet{lopezlira2026} derive optimal fees and capacity under AI-driven submission growth. Desk rejections are increasingly common in economics \citep{card2020}. These papers study fees, delays or editorial screening rather than the fixed, non-transferable quota of submission rights I analyse below.

Related research examines alternative mechanisms for journal submissions and information asymmetry in editorial decisions. \citet{pruferzetland2010} propose an auction mechanism whereby editors bid for papers they want to publish; \citet{azar2015} models a review process in which authors know manuscript quality while referees receive a noisy signal; \citet{pernagallo2023} studies asymmetric information between researchers and editors; \citet{pernagallo2025} studies scientific publication under submission fees and free submission. In this paper, I consider a different allocation mechanism: non-transferable submission quotas with a common fee at a fixed expected review load.

Recent work and policy have focused on quota implementation under co-authorship. \citet{shah2026} develops a harmonic quota; \citet{tmlr2026} adopt a generalised harmonic rule. \citet{mumcuyilmaz2026} propose quotas that assign role-specific quota costs. The fee--quota comparison belongs to the classic prices-versus-quantities tradition \citep{weitzman1974}. Proposition 1 below is analogous to the result that rationing can be preferable to prices when willingness to pay is a poor proxy for value \citep{weitzman1977, dworczak2021, akbarpour2024}. Here, each author selects manuscripts to submit from their privately ranked set of candidates. Hence a quota combines the between-author allocation of review rights with within-author curation. The comparison between quota and fee depends on the dispersion of fee-induced selectivity across authors and the relationship between that selectivity and baseline portfolio quality. Related ideas appear in work on linked-decision mechanisms \citep{jackson2007}, ranking and quota mechanisms \citep{frankel2014}, non-monetary allocation \citep{hylland1979, budish2011, prendergast2022}, and preference signalling in matching markets \citep{coles2013}.

\section{Model and analysis}

I study a partial-equilibrium model of manuscript submission. An author has $M$ candidate manuscripts with values $V_1,\ldots,V_M$. The $V_j$ are independently and identically distributed (i.i.d.) according to a continuous distribution $F$ that is supported on $[0,\infty)$ with finite mean. I denote by $g(V_j)$ the author's perceived return from manuscript $j$, where $g$ is strictly increasing and non-negative. I assume that the author can correctly rank the $V_j$. Since $g$ is strictly increasing, it follows that ranking manuscripts by perceived return $g(V_j)$ is equivalent to ranking them by $V_j$. I do not require that $g(V_j)=V_j$ so the author is not necessarily cardinally calibrated. There is a single research platform. It grants the author $K<M$ non-transferable submission rights. It follows from $g \ge 0$ that the quota is almost surely binding. Let $V_{(1:M)}\leq\cdots\leq V_{(M:M)}$ denote the order statistics. The average value of the $K$ highest-valued candidate manuscripts is given by
\[
Q_K(M)=\frac{1}{K}\sum_{r=M-K+1}^{M}V_{(r:M)}.
\]
Each $V_j$ may also be interpreted as the author's private signal of the underlying scientific value $S_j$. Let $m(v) = \E(S_j \mid V_j=v)$ denote the expected scientific value conditional on the signal. If the pairs $(V_j, S_j)$ are independently and identically distributed across $j$ then $\E(S_j \mid V_1=v_1, \ldots, V_M=v_M) = \E(S_j \mid V_j = v_j) = m(v_j)$, so the expected scientific value of manuscript $j$ only depends on its own signal $V_j$, even after seeing the whole candidate set. Suppose also that $m$ is non-decreasing in $v$, so that candidate manuscripts with higher $V_j$ also have weakly higher expected scientific value. The author does not need to observe $S_j$. Under this interpretation, the selection results in Lemma 1 carry over to expected scientific value with weak inequalities. The fee--quota ranking in Proposition 1 is preserved under a common affine transformation $m(v)=a+bv$, with $b>0$, but may not be preserved under a general non-decreasing $m$. 

\begin{lemma}[Curation under a binding quota]
Suppose that the author has $M$ candidate manuscripts and a binding quota of $K<M$ submissions.
\begin{enumerate}
\item The author selects the $K$ candidates with the highest values $V_j$. Since $g$ is strictly increasing, these are also the $K$ candidates with the highest perceived returns $g(V_j)$. Hence the selected set is invariant to the author's cardinal calibration.

\item For any realised candidate set, adding a candidate manuscript cannot reduce the total, and hence the average, value of the selected manuscripts.

\item If the candidate values are i.i.d.\ from a continuous distribution, then adding a candidate strictly increases the expected average value of the selected portfolio:
\[
\E Q_K(M+1)>\E Q_K(M).
\]

\item If the candidate values are i.i.d. uniform on $[0,1]$, then
\[
\E Q_K(M)=1-\frac{K+1}{2(M+1)}.
\]
\end{enumerate}
\end{lemma}

\noindent\emph{Proof.}
Suppose that a selected candidate has a lower value $V_j$ than one that is not selected. Because $g$ is strictly increasing, switching these two manuscripts will increase the author's total perceived return. This implies that a binding quota will select the $K$ highest-value candidates, and any strictly increasing transformation of the values will leave the selected set unchanged. Now add one candidate to a realised set of $M$ candidates. There are two cases. If the new candidate does not enter the top $K$, the selected set is unchanged. If it does enter, it will replace the previous $K$th-highest candidate. This establishes that the total, and hence the average, value of the selected manuscripts cannot fall. Under i.i.d.\ sampling from a continuous distribution, the new candidate is among the top $K$ of the resulting $M+1$ candidates with positive probability $K/(M+1)$. On this event, both the total and average value of the selected portfolio strictly increase. Otherwise, the selected portfolio is unchanged. It follows that
\[
\E Q_K(M+1)>\E Q_K(M).
\]
Finally, if the candidate values are i.i.d.\ uniform on $[0,1]$, then
$\E V_{(r:M)}=r/(M+1)$, so
\[
\E Q_K(M)
 =\frac{1}{K}\sum_{r=M-K+1}^{M}\frac{r}{M+1}
 =1-\frac{K+1}{2(M+1)}.
\]
\hfill$\square$

In the absence of rationing, an increase in the number of candidates can generate additional submissions and raise the platform's review load. On the other hand, subject to a binding quota, each active author will still submit exactly $K$ manuscripts, but they will now choose them from a larger candidate set. Consequently, the selected portfolio cannot worsen and will indeed improve in expectation under the i.i.d.\ assumption of Lemma 1. Sequential submissions may yield smaller gains. This is not a claim that AI raises scientific quality. Additional candidates may be weak. What matters here is that authors' rankings are sufficiently informative for them to reserve their submission rights for candidates with higher expected scientific value.

Now consider $N$ active authors. The $i$th author has $M>K$ candidate manuscripts with values $V_{ij}=\theta_i+U_{ij}$ where the $U_{ij}$ are i.i.d. uniform on $[0,1]$. In this framework, $\theta_i\geq0$ represents a baseline portfolio-quality parameter, while $U_{ij}$ is a manuscript-specific component of value. Subject to a common fee $p$, author $i$ will submit candidate $j$ whenever $g_i(V_{ij})\geq \mu_i p$, where $g_i(V_{ij})$ is that author's perceived return and $\mu_i > 0$ is a reduced-form fee sensitivity that converts the fee into the effective cost of submission. A larger $\mu_i$ makes a given fee more restrictive, holding $g_i$ fixed.

Assume that $g_i\geq0$ is continuous and strictly increasing. The submission rule can then be written as $U_{ij}\geq d_i$, where $d_i\in[0,1]$ is an effective cut-off. If $\mu_i p\leq g_i(\theta_i)$, then $d_i=0$ and all candidate manuscripts are submitted. If $g_i(\theta_i)<\mu_i p<g_i(\theta_i+1)$, the cut-off is interior and satisfies $g_i(\theta_i+d_i)=\mu_i p$. If $\mu_i p\geq g_i(\theta_i+1)$, then $d_i=1$ and the submission probability is zero. Equivalently, let $c_i=\theta_i+d_i$. Then author $i$ submits manuscript $j$ if and only if $V_{ij}\geq c_i$. It is useful to interpret $d_i$ as the within-author cut-off, while $c_i$ is the corresponding cut-off in manuscript value.

I compare the fee and quota at equal expected review load, treating fee revenue as a transfer. Under the fee, the expected review load is $M\sum_{i=1}^N(1-d_i)$ and the quota generates exactly $NK$ submissions. The expected review load is continuous in $p$. Furthermore, expected load is $NM$ at $p=0$ and zero for sufficiently large $p$. This guarantees that, for any intermediate expected load, there exists a fee that generates it, including $NK$. It follows that the equal-load fee satisfies
\[
\bar d\equiv \frac{1}{N}\sum_{i=1}^N d_i = 1-\frac{K}{M}.
\]
Since $1-d_i$ is author $i$'s expected submission rate, a higher $\bar d$ corresponds to greater average fee-induced selectivity and a lower average submission rate. In the equal-load comparison, $\bar d$ is fixed and the average expected submission rate is identical under the fee and the quota.

Even though the equal expected review load determines the average fee-induced cut-off, the individual cut-offs $d_i$ may still differ substantially across authors. Let $\bar\theta$ denote the mean baseline portfolio quality, and let $\sigma_d^2$ denote the cross-author variance of $d_i$. Finally, denote by $\gamma=\Cov(\theta,d)$ the cross-author covariance between baseline portfolio quality and the fee-induced cut-off. All cross-author moments are defined using denominator $N$. Under both the quota and the fee, define mean reviewed value as expected total reviewed value divided by expected review load. Let $\sigma_c^2$ and $\sigma_\theta^2$ denote the cross-author variances of $c_i$ and $\theta_i$, respectively.

\begin{proposition}[Quotas versus equal-load fees]
Under the uniform benchmark above, mean reviewed values under the $K$-quota and the equal-load fee are, respectively,
\[
q_Q=\bar\theta+1-\frac{K+1}{2(M+1)},\qquad
q_F=\bar\theta+1-\frac{K}{2M}-\frac{M}{2K}\left(\sigma_d^2+2\gamma\right).
\]
Hence the quota yields weakly higher mean reviewed value if and only if
\[
\sigma_d^2+2\gamma\geq \tau(M,K)\equiv\frac{K(M-K)}{M^2(M+1)}.
\]
Equivalently, since $c_i=\theta_i+d_i$, the condition is $\sigma_c^2-\sigma_\theta^2\geq\tau(M,K)$.
\end{proposition}

\noindent\emph{Proof.}
Under the quota, each author submits the $K$ highest-valued candidates. Since $\theta_i$ shifts all of author $i$'s candidate values by the same amount, the quantity $q_Q$ then follows immediately from Lemma 1. Under the fee, author $i$ submits whenever $U_{ij}\geq d_i$. The expected total value of this author's submitted manuscripts is therefore
\[
M\int_{d_i}^1(\theta_i+u)\,du
=
M\left\{(1-d_i)\theta_i+\frac{1-d_i^2}{2}\right\}.
\]
Summing across authors and dividing by the equal expected review load $NK$
gives
\[
q_F
=
\frac{M}{K}\left[
\bar\theta-\frac{1}{N}\sum_i\theta_i d_i
+\frac{1}{2}\left\{1-\frac{1}{N}\sum_i d_i^2\right\}
\right].
\]
Substituting
\[
\frac{1}{N}\sum_i\theta_i d_i=\bar\theta\bar d+\gamma,
\qquad
\frac{1}{N}\sum_i d_i^2=\bar d^2+\sigma_d^2,
\]
and the equal-load condition $\bar d=1-K/M$, it follows that
\[
q_F = \bar\theta+1-\frac{K}{2M}
-\frac{M}{2K}\left(\sigma_d^2+2\gamma\right).
\]
Hence $q_Q\geq q_F$ if and only if
\[
\sigma_d^2+2\gamma
\geq
\frac{K(M-K)}{M^2(M+1)}.
\]
Finally, since $c_i=\theta_i+d_i$,
\[
\sigma_c^2 = \sigma_\theta^2+\sigma_d^2+2\gamma,
\]
which gives the equivalent variance condition.
\hfill$\square$

Proposition~1 identifies two forces that can offset the flexibility advantage of a common fee. First, when the mean cut-off and the covariance term $\gamma$ are fixed, greater dispersion in the cut-offs $d_i$ will lower the mean reviewed value. Authors whose cut-offs are low will submit more manuscripts in expectation, including candidates that are further down their own rankings. On the other hand, authors with high cut-offs will submit fewer manuscripts. Second, $\gamma$ describes whether this variation in selectivity is systematically related to baseline portfolio quality. If $\gamma>0$, authors with higher $\theta_i$ will also have higher cut-offs and therefore submit fewer manuscripts, shifting review capacity away from stronger portfolios. If $\gamma<0$, stronger portfolios receive more review capacity and the fee performs better.

At the same time, a common threshold allows the number of submissions to vary with the realised quality of an author's candidate set, whereas a binding quota fixes that number at $K$. Hence, when $\sigma_d^2=0$, the fee performs strictly better because it can accommodate unusually strong or weak realised portfolios. The threshold $\tau(M,K)$ measures how large the combined effect of cut-off dispersion and its association with baseline portfolio quality must be to offset this flexibility advantage. It tends to zero as $M$ grows with $K$ fixed. This does not imply that quotas necessarily become more attractive: equal expected review load also implies $\bar d=1-K/M\to1$, leaving less scope for dispersion in the $d_i$. It is better to examine the limiting effect of manuscript abundance directly: I do this in Corollary~1 by allowing fee-induced cut-offs to attain the boundary.

A simple interior case gives the threshold in Proposition~1 a direct interpretation in terms of heterogeneity in fee sensitivity. Suppose that $\theta_i=0$ and $g_i(v)=v$ and all fee-induced cut-offs are interior. Then $d_i=\mu_i p$. Combining this relation with the equal-load condition implies that the quota weakly outperforms the fee if and only if
\[
\CV(\mu)^2 \ge \frac{K}{(M-K)(M+1)},
\]
where $\CV(\mu)=\sigma_\mu/\bar\mu$ is the coefficient of variation of fee sensitivity. The comparison is expressed in terms of cross-author heterogeneity in fee sensitivity. For example, when $M=10$ and $K=2$, the critical coefficient of variation is approximately $0.151$. The interior restriction is
\[
\frac{\max_i\mu_i}{\bar\mu}<\frac{M}{M-K}.
\]
This restriction has a simple implication as $M$ becomes large. If fee sensitivities are not identical, then $\max_i\mu_i / \bar\mu > 1$, whereas $M / (M-K) \to1$ as $M$ grows with $K$ fixed. The interior condition must eventually fail and the most fee-sensitive authors then reach $d_i=1$ and submit no manuscripts. Proposition~1 still applies; only the interior relation $d_i=\mu_i p$ ceases to hold for those authors.

\begin{corollary}[Abundance can turn fees into an access screen]
Let $\theta_i=0$ and $g_i(v)=v$ for all $i$, and let the $\mu_i$ be fixed, with a share $s\in(0,1)$ of authors at $\mu_{\min}$ and $r=\mu_{(2)}/\mu_{\min}>1$, where $\mu_{(2)}$ is the next-lowest distinct value. Under the equal-load fee, $d_i=\min\{\mu_i p,1\}$. If $M\geq Kr/\{s(r-1)\}$, every author with $\mu_i>\mu_{\min}$ is priced out, $q_F=1-K/(2sM)$, and the quota yields strictly higher mean reviewed value if and only if $s(K+1)M<K(M+1)$; in particular, whenever $s\leq1/2$.
\end{corollary}

\noindent\emph{Proof.} Suppose first that only the $\mu_{\min}$ group is active. Equal expected review load then requires $sNM(1-d)=NK$ so $d=1-K/(sM)$, hence $p=d/\mu_{\min}$. The next-lowest fee-sensitivity group is priced out when $\mu_{(2)}p\geq1$, that is, when $M\geq Kr/\{s(r-1)\}$. All authors with larger $\mu_i$ are then priced out too. For the remaining active authors, submitted values are uniform on $[d,1]$, so $q_F=(1+d)/2=1-K/(2sM)$, while $q_Q=1-(K+1)/\{2(M+1)\}$ by Lemma 1. Hence $q_Q>q_F$ if and only if $s(K+1)M<K(M+1)$. If $s\leq1/2$, then $s(K+1)M\leq(K+1)M/2<K(M+1)$, which proves the final claim. \hfill$\square$

Corollary~1 describes an access effect. As $M$ increases, higher fee-induced cut-offs are needed to maintain the same expected review load. If manuscript values are bounded, and authors are sufficiently fee-sensitive, these cut-offs can eventually reach the upper boundary. Those authors then have zero submission probability and are priced out entirely. Under a quota, by contrast, each active author retains $K$ submission rights. The exclusion threshold depends on the share $s$ of least fee-sensitive authors rather than on the number of authors $N$. For example, if $K=2$, $r=1.5$, and $s=1/2$, then from $M=12$ onward only the least fee-sensitive group submits under the equal-load fee. The sharp transition, and the resulting closed-form threshold, is a direct consequence of the gap between the mass point at $\mu_{\min}$ and the next-lowest type. With smoother heterogeneity in fee sensitivity, participation would decline more gradually. Complete exclusion also requires bounded manuscript values. With unbounded values, sufficiently fee-sensitive authors would typically submit less, but would retain a positive probability of drawing a manuscript valuable enough to clear the fee. The result concerns mean reviewed value at a fixed expected review load, rather than general welfare. As $M\to\infty$, with $K$ and $s$ fixed, it follows that $q_Q-q_F\to0$, although more fee-sensitive authors remain excluded under the fee.

\section{Design implications}

The results give some insight into the design of effective submission policies. Lemma~1 shows that quotas create an opportunity cost. Specifically, using a submission right for one candidate prevents an author from using it for another manuscript. There are parallels in research on ranking and budget mechanisms and linked-decision mechanisms \citep{frankel2014, jackson2007}. The lemma provides some justification for the suggestion that quotas should be applied to first-round submissions rather than acceptances, since an acceptance-based rule does not require authors to prioritise before submission. Invited manuscript revisions can be treated separately because the platform has already chosen to continue reviewing the manuscript. The accounting unit matters for the same reason. A submitter-level rule can be weakened when co-authors rotate submitters. In addition, charging every co-author fully may discourage collaboration. In principle, harmonic accounting rules may more directly preserve the scarcity of submission rights \citep{shah2026,tmlr2026}.

In Proposition~1, fees and quotas are seen to act on different margins. They are not necessarily substitutes. A fee changes participation and submission thresholds across authors. A quota instead limits the maximum number of submissions. A combination may have some merit: a fee could be used to limit submission intensity, while a quota avoids extreme submission volume by an individual author and forces authors to prioritise their work \citep{cmr2026,tar2026}. Corollary~1 shows that when manuscript values are bounded and there is an abundance of manuscripts, an equal-load fee can price out sufficiently fee-sensitive authors altogether. By contrast, a quota reserves submission rights across active authors. A quota directly constrains the intensive margin, while the fee can also affect the extensive margin by inducing some authors not to submit. Neither mechanism, however, directly controls the number of active authors in the system. Aggregate submissions can still rise when the number of active authors increases, or when manuscripts are redirected to uncapped platforms. Furthermore, the model takes candidate sets as given, although it is plausible that submission rules may affect research effort, manuscript development, bundling, and project choice. Related work by \citet{ellison2002standards} shows more generally that publishing institutions and standards can affect how manuscripts are developed.

\section{Conclusion}

In this paper, I study the problem of rationing expert review capacity when manuscript production is cheap. Submission fees and non-transferable quotas use different information to address the issue. A common fee uses an author's willingness to pay, and allows submission volume to vary across authors and realised candidate sets. A quota, however, gives an author a fixed number of submission rights and makes the author prioritise and submit their highest-ranked manuscripts. A quota is less flexible than a fee; its benefit is that allocation depends on ordinal ranking information within an author's set of candidate manuscripts, even when cardinal assessments and fee sensitivities vary across authors. When manuscript values are bounded, the distinction can become sharper still, because a sufficiently high fee can price out some authors from the submission pool entirely.

\clearpage

{\footnotesize
\setlength{\bibsep}{0pt}
\bibliographystyle{elsarticle-harv}
\bibliography{ms}
}

\clearpage

\medskip\noindent\textbf{Declaration of generative AI and AI-assisted technologies in the manuscript preparation process.}
During the preparation of this work, the author used OpenAI ChatGPT (GPT-5.6 Sol) and Anthropic Claude (Opus 5.5) to assist with parts of the literature search. All sources suggested by these tools were independently reviewed and verified by the author. The author takes full responsibility for the content of the publication.

\clearpage

\ifwithappendix
  \clearpage
  \begingroup


\setcounter{section}{0}
\setcounter{subsection}{0}
\setcounter{equation}{0}

\renewcommand{\thesection}{\Alph{section}}
\renewcommand{\thesubsection}{\thesection.\arabic{subsection}}
\renewcommand{\theequation}{\thesection.\arabic{equation}}

\makeatletter
\@addtoreset{equation}{section}
\makeatother


\begin{center}
    {\LARGE Supplementary Appendix to\\[0.4em]
    \emph{Submission quotas as self-screening\\
    when manuscript production is cheap}\par}

    \vspace{1.2em}

    \renewcommand{\thefootnote}{\fnsymbol{footnote}}
    \setcounter{footnote}{0}
    {\large Steve Lawford\footnote{ENAC (University of Toulouse),
    7 avenue Edouard Belin, BP 54005, 31055 Toulouse Cedex 4,
    France (email: steve.lawford@enac.fr).}\par}

\end{center}

\renewcommand{\thefootnote}{\arabic{footnote}}
\setcounter{footnote}{0}

\vspace{1.5em}

\noindent This appendix provides detailed derivations, proofs, and additional discussion for the paper \emph{Submission quotas as self-screening when manuscript production is cheap}. The main article contains the statements and proofs of all principal results and can be read independently.

\section{Notation and assumptions}\label{sec:notation}

\subsection{Single-author curation problem}

An author has $M$ candidate manuscripts with values
\[
V_1,\ldots,V_M,
\]
drawn i.i.d. from a common continuous distribution $F$ supported on $[0,\infty)$ with finite mean. The author does not have to assign the correct cardinal value to a manuscript, but instead perceives a return $g(V_j)$, where $g$ is a strictly increasing and non-negative function. A platform permits the author to make at most $K<M$ non-transferable submissions.

Because $g\geq 0$, and there is no submission cost, an additional submitted manuscript can never lower an author's perceived return. Except for some zero-probability boundary cases, the author will therefore use all of their $K$ submission rights. The strict monotonicity of $g$ implies that the ordering of $g(V_j)$ is identical to the ordering of $V_j$, and so the quota only requires an ordinal ranking of the candidate set.

The order statistics are denoted by
\[
V_{(1:M)}\leq\cdots\leq V_{(M:M)},
\]
and the average value of the selected top-$K$ portfolio is defined by
\begin{equation}
Q_K(M)
=\frac{1}{K}\sum_{r=M-K+1}^{M}V_{(r:M)}.
\label{eq:qkm_def}
\end{equation}

Ties are ruled out by the assumption of continuity. When this is combined with i.i.d. draws, it follows that expected portfolio value will strictly improve as the candidate set becomes larger. Neither continuity nor i.i.d. draws are used in proving the core top-$K$ selection result.

\subsection{Signal interpretation}

The variable $V_j$ need not be a manuscript's literal underlying scientific value $S_j$. Instead, it can be interpreted as an author's private signal of that scientific value. Let $S_j$ be integrable and define
\[
m(v)=\E(S_j\mid V_j=v).
\]
Suppose the pairs $(V_j,S_j)$ are i.i.d.\ across candidates. The quota  uses the author's ranking of signals. The author need not observe $S_j$ directly.

The i.i.d. assumption on the pairs implies
\begin{equation}
\E(S_j\mid V_1,\ldots,V_M)=\E(S_j\mid V_j)=m(V_j).
\label{eq:signal_conditional}
\end{equation}
It follows that, conditional on the vector of private signals, the expected scientific value of a selected set $A$ determined by the private signals is given by
\begin{equation}
\E\left(\sum_{j\in A}S_j\,\middle|\,V_1,\ldots,V_M\right)
=\sum_{j\in A}m(V_j).
\label{eq:signal_sum}
\end{equation}
When $m$ is non-decreasing, selecting the largest $V_j$ weakly maximises expected scientific value. Strictness need not hold because $m$ may have regions that are flat, even when the underlying signals themselves are distinct.

In Proposition 1, suppose that the conditional expected scientific value is given by a common affine function $m(v)=a+bv$, with $b>0$. Since the quota and fee have the same expected review load, their mean expected scientific values satisfy $q_Q^S=a+bq_Q$ and $q_F^S=a+bq_F$. Hence $q_Q^S-q_F^S=b(q_Q-q_F)$, so the fee--quota ranking in Proposition 1 is preserved. This may not hold for a general non-decreasing function $m$.

\subsection{Multi-author fee--quota benchmark}

There are $N$ active authors, each of whom has the same number $M>K$ of candidate manuscripts, with
\begin{equation}
V_{ij}=\theta_i+U_{ij},
\qquad
U_{ij}\stackrel{\mathrm{i.i.d.}}{\sim}U[0,1],
\label{eq:benchmark_values}
\end{equation}
where $\theta_i\geq 0$ is a baseline portfolio-quality parameter and $U[0,1]$ is the uniform distribution.

Author $i$ perceives $g_i(V_{ij})$, where $g_i$ is non-negative, continuous, and strictly increasing. Author $i$ has a fee sensitivity component $\mu_i>0$. Under a common submission fee $p$, candidate $j$ is submitted when
\begin{equation}
g_i(V_{ij})\geq \mu_i p.
\label{eq:fee_rule}
\end{equation}
Authors differ in both manuscript valuations and fee sensitivity. The function $g_i(V_{ij})$ captures the former, while $\mu_i$ captures the latter.

Let the effective cut-off $d_i\in[0,1]$ in $U$-space be defined by
\begin{equation}
d_i(p)=
\begin{cases}
0, & \mu_i p\leq g_i(\theta_i),\\[3pt]
g_i^{-1}(\mu_i p)-\theta_i, & g_i(\theta_i)<\mu_i p<g_i(\theta_i+1),\\[3pt]
1, & \mu_i p\geq g_i(\theta_i+1).
\end{cases}
\label{eq:di_definition}
\end{equation}
The middle expression is well defined because $g_i$ is continuous and strictly increasing. Almost surely, author $i$ submits exactly those candidates for which
\[
U_{ij}\geq d_i.
\]
The corresponding cut-off in $V$-space is then
\begin{equation}
c_i=\theta_i+d_i.
\label{eq:ci_def}
\end{equation}

For later reference, define cross-author means, variances, and covariance by
\begin{equation}
\bar d=\frac{1}{N}\sum_{i=1}^N d_i,
\qquad
\bar\theta=\frac{1}{N}\sum_{i=1}^N\theta_i,
\label{eq:cross_means}
\end{equation}

\begin{equation}
\sigma_d^2=\frac{1}{N}\sum_{i=1}^N(d_i-\bar d)^2,
\qquad
\sigma_\theta^2=\frac{1}{N}\sum_{i=1}^N(\theta_i-\bar\theta)^2,
\label{eq:cross_vars}
\end{equation}

\begin{equation}
\gamma=\Cov(\theta,d)
=\frac{1}{N}\sum_{i=1}^N(\theta_i-\bar\theta)(d_i-\bar d).
\label{eq:gamma_def}
\end{equation}
These are cross-author population moments not sampling moments.

Mean reviewed value is defined by expected total reviewed value divided by expected review load. Under the quota, the load is fixed at $NK$, whereas under the fee, it is random.

\subsection{Role of the main assumptions}

The assumptions have simple roles:
\begin{itemize}
\item \textbf{Continuity} rules out ties; combined with i.i.d. draws, it ensures that the expected portfolio value strictly improves when the candidate set grows.
\item \textbf{Monotonicity of $g$} preserves manuscript rankings and is used in the threshold submission rules.
\item \textbf{i.i.d.\ sampling} is required for the order-statistic results, and i.i.d. \emph{pairs} $(V_j,S_j)$ further support the signal results.
\item \textbf{Uniform $U_{ij}$} facilitates the closed-form expressions that appear in Proposition~1 but it is not required for the top-$K$ curation mechanism to hold.
\item \textbf{Common $M$ and binding $K<M$} are used to give a quota review load of exactly $NK$.
\item \textbf{Bounded support $U\in[0,1]$} is essential for the pricing out result of Corollary~1.
\item \textbf{Equal expected review load} holds the expected use of review capacity fixed across the quota and the fee.
\item \textbf{Fee revenue treated as a transfer} means that fee payments do not carry a social cost; this paper does not develop a general welfare analysis.
\end{itemize}

\section{Detailed proof of Lemma 1: curation under a binding quota}\label{sec:lemma}

The main paper states the following result.

\begin{lemmaone}
Suppose that the author has $M$ candidate manuscripts and a binding quota of $K<M$ submissions.
\begin{enumerate}
\item The author selects the $K$ candidates with the highest values $V_j$. Since $g$ is strictly increasing, these are also the $K$ candidates with the highest perceived returns $g(V_j)$. Hence the selected set is invariant to the author's cardinal calibration.

\item For any realised candidate set, adding an additional candidate manuscript cannot reduce the total, and hence the average, value of the selected manuscripts.

\item If the candidate values are i.i.d.\ from a continuous distribution, then adding a candidate strictly increases the expected average value of the selected portfolio:
\[
\E Q_K(M+1)>\E Q_K(M).
\]

\item If the candidate values are i.i.d. uniform on $[0,1]$, then
\[
\E Q_K(M)=1-\frac{K+1}{2(M+1)}.
\]
\end{enumerate}
\end{lemmaone}

\subsection{Top-K selection and invariance to cardinal transformations}

Let $A$ denote a feasible set of $K$ submitted candidates. Since each selected candidate gives perceived return $g(V_j)\geq0$, a binding quota leads the author to choose a set that maximises
\begin{equation}
\sum_{j\in A}g(V_j)
\qquad\text{subject to}\qquad |A|=K,
\label{eq:author_objective}
\end{equation}
where $|A|$ is the size of set $A$. Suppose that a maximising set $A$ includes candidate $a$ but excludes candidate $b$, with $V_b>V_a$. Then, strict monotonicity gives $g(V_b)>g(V_a)$. Replacing $a$ by $b$ changes the objective by
\[
g(V_b)-g(V_a)>0,
\]
which contradicts optimality. It follows that no selected candidate can lie below an unselected candidate. With probability one, the selected set is exactly the $K$ candidates with the highest values,
\[
\{V_{(M-K+1:M)},\ldots,V_{(M:M)}\}.
\]

Using the same argument, we can show invariance to any strictly increasing transformation. If $\tilde g$ is strictly increasing, then for every pair $a$ and $b$,
\[
V_a>V_b
\quad\Longleftrightarrow\quad
\tilde g(V_a)>\tilde g(V_b).
\]
Hence every strictly increasing cardinal transformation induces the same ordering and the same selected set. The quota needs the author to know which candidates are better than others, but not by how much.

\subsection{Adding a candidate cannot worsen the selected portfolio}

Fix an arbitrary realised set of $M$ values and let
\[
x_{[1]}\geq x_{[2]}\geq\cdots\geq x_{[M]}
\]
denote them in descending order. The sum selected under the $K$-quota is
\[
T_K(M)=\sum_{r=1}^{K}x_{[r]}.
\]
Now add a candidate with realised value $x$. If $x\leq x_{[K]}$, the new candidate will not enter into the top $K$ manuscripts, and
\[
T_K(M+1)=T_K(M).
\]
If $x>x_{[K]}$, the new candidate will enter the top $K$, thus displacing the old $K$th-best candidate, so
\[
T_K(M+1)=T_K(M)+x-x_{[K]}>T_K(M).
\]
Hence, for every realised candidate set,
\begin{equation}
T_K(M+1)\geq T_K(M).
\label{eq:deterministic_monotonicity}
\end{equation}
In other words, each relevant upper order statistic is weakly non-decreasing when the candidate set is enlarged. Dividing by $K$ gives
\begin{equation}
Q_K(M+1)\geq Q_K(M)
\label{eq:q_monotone_pathwise}
\end{equation}
where the first $M$ observations are kept and one new observation is added.

\subsection{Strict improvement in expectation under i.i.d. sampling}

Denote by $V_{[K]}^{\downarrow}$ the $K$th-largest realised value, conditional on the first $M$ draws. The new draw $V_{M+1}$ will enter an author's top $K$ set of manuscripts whenever
\[
V_{M+1}>V_{[K]}^{\downarrow}.
\]
Given a continuous non-degenerate distribution, this event has (almost surely) positive conditional probability. Using \eqref{eq:deterministic_monotonicity}, it follows that 
\begin{equation}
\E[T_K(M+1)-T_K(M)]>0,
\label{eq:strict_topk_gain}
\end{equation}
where the assumption of finite mean ensures integrability. Dividing \eqref{eq:strict_topk_gain} by $K$ gives
\begin{equation}
\E Q_K(M+1)>\E Q_K(M).
\label{eq:strict_q_gain}
\end{equation}

\subsection{Uniform closed form}

If $V_j\sim U[0,1]$, the $r$th ascending order statistic follows a Beta distribution,
\begin{equation}
V_{(r:M)}\sim\Beta(r,M+1-r),
\label{eq:beta_order_stat}
\end{equation}
and immediately
\begin{equation}
\E V_{(r:M)}=\frac{r}{M+1}.
\label{eq:uniform_order_mean}
\end{equation}
Therefore
\begin{equation}
\E Q_K(M) = \frac{1}{K}\sum_{r=M-K+1}^{M}\frac{r}{M+1} = \frac{1}{K(M+1)}\sum_{r=M-K+1}^{M}r.
\label{eq:q_uniform_sum}
\end{equation}
The sum contains $K$ consecutive integers. Hence
\begin{equation}
\sum_{r=M-K+1}^{M}r
=\frac{K\{(M-K+1)+M\}}{2}
=\frac{K(2M-K+1)}{2}.
\label{eq:arith_sum}
\end{equation}
Substituting into \eqref{eq:q_uniform_sum} yields
\begin{equation}
\E Q_K(M) =\frac{2M-K+1}{2(M+1)} =1-\frac{K+1}{2(M+1)}.
\label{eq:q_uniform_final}
\end{equation}

\subsection{Expected scientific value under informative private signals}

Adding one candidate cannot reduce the sum of the top $K$ values of $m(V_j)$ and so taking expectations gives a weak curation result for underlying scientific value; a new candidate can have a higher signal while leaving $m(V)$ unchanged when the conditional-mean function is locally flat.

The author need not observe a manuscript's true scientific value, but they must have an informative ranking about it. Under a purely noisy ranking, the quota still limits volume but the curation mechanism is lost.

\section{The equal-load fee benchmark}\label{sec:equal_load}

\subsection{Expected load under a fee}

For author $i$, a candidate is submitted when $U_{ij}\geq d_i$. Since $U_{ij}\sim U[0,1]$,
\begin{equation}
\Pr(U_{ij}\geq d_i)=1-d_i.
\label{eq:submit_prob}
\end{equation}
With $M$ candidates, the expected number submitted by author $i$ is
\begin{equation}
L_i(p)=M(1-d_i(p)).
\label{eq:author_load}
\end{equation}
This gives expected total review load under the fee equal to
\begin{equation}
L(p)=M\sum_{i=1}^N\{1-d_i(p)\}.
\label{eq:total_load}
\end{equation}
At $p=0$, it follows that $d_i=0$, so $L(0)=NM$. When $p$ is sufficiently large, every $d_i=1$, and so $L(p)=0$. By continuity of each ``clipped'' threshold $d_i(p)$, $L(p)$ is continuous. From
\[
0<NK<NM,
\]
the intermediate value theorem then establishes that at least one fee $p$ generates expected load $NK$.

\subsection{The equal-load restriction}

Choose such a fee. Equating \eqref{eq:total_load} to $NK$ gives
\begin{equation}
M\sum_{i=1}^N(1-d_i) = 
MN-M\sum_{i=1}^N d_i = NK.
\label{eq:load_equal_1}
\end{equation}
Dividing by $MN$,
\begin{equation}
1-\frac{1}{N}\sum_{i=1}^N d_i=\frac{K}{M},
\label{eq:load_equal_2}
\end{equation}
and it follows that
\begin{equation}
\boxed{\bar d=1-\frac{K}{M}.}
\label{eq:bar_d}
\end{equation}
This is a central result. It fixes the mean threshold across authors but does not fix its dispersion nor its covariance with baseline portfolio quality. Proposition~1 shows that the remaining cross-author features determine a fee's performance relative to a quota.

\section{Detailed proof of Proposition 1}\label{sec:prop}

The main paper states the following result.

\begin{propositionone}
Under the uniform benchmark above, mean reviewed values under the $K$-quota and the equal-load fee are, respectively,
\[
q_Q=\bar\theta+1-\frac{K+1}{2(M+1)},
\qquad
q_F=\bar\theta+1-\frac{K}{2M}-\frac{M}{2K}(\sigma_d^2+2\gamma).
\]
Hence the quota yields weakly higher mean reviewed value if and only if
\[
\sigma_d^2+2\gamma\geq
\tau(M,K)
\equiv
\frac{K(M-K)}{M^2(M+1)}.
\]
Equivalently, since $c_i=\theta_i+d_i$, the condition is $\sigma_c^2-\sigma_\theta^2\geq\tau(M,K)$.
\end{propositionone}

\subsection{Mean reviewed value under the quota}

Author $i$ has candidate values $\theta_i+U_{ij}$. Adding a constant $\theta_i$ does not alter the within-author ranking. Consequently, the $K$ selected candidates have average value
\begin{equation}
\theta_i+Q_{K,U}(M),
\label{eq:quota_author_quality}
\end{equation}
where $Q_{K,U}(M)$ is the average of the top $K$ of the $M$ uniform $U$ draws. By \eqref{eq:q_uniform_final},
\begin{equation}
\E Q_{K,U}(M)
=1-\frac{K+1}{2(M+1)}.
\label{eq:u_topk}
\end{equation}
Because every active author contributes exactly $K$ reviewed manuscripts, mean reviewed value under the quota is the equally weighted average of the author-specific expected means:
\begin{equation}
q_Q =\frac{1}{N}\sum_{i=1}^N
\left\{\theta_i+1-\frac{K+1}{2(M+1)}\right\}
 = \boxed{\bar\theta+1-\frac{K+1}{2(M+1)}}.
\label{eq:qQ_detailed}
\end{equation}

\subsection{Expected reviewed value under the fee}

For a given candidate of author $i$, the reviewed value is
\[
(\theta_i+U_{ij})\1\{U_{ij}\geq d_i\}.
\]
Since $U_{ij}$ has density one on $[0,1]$,
\begin{equation}
\E\left[(\theta_i+U_{ij})\1\{U_{ij}\geq d_i\}\right]
=\int_{d_i}^{1}(\theta_i+u)\,du
= \theta_i(1-d_i)+\frac{1-d_i^2}{2}.
\label{eq:one_candidate_value}
\end{equation}
With $M$ candidates, the expected total reviewed value from author $i$ is
\begin{equation}
R_i^F
=M\left\{\theta_i(1-d_i)+\frac{1-d_i^2}{2}\right\}.
\label{eq:author_review_value}
\end{equation}
Summing this expression across authors,
\begin{equation}
R^F
=M\sum_{i=1}^N
\left\{\theta_i(1-d_i)+\frac{1-d_i^2}{2}\right\}.
\label{eq:total_review_value_fee}
\end{equation}
The equal-load fee generates expected review load $NK$, so the mean reviewed value is
\begin{equation}
q_F
=\frac{R^F}{NK}
= \frac{M}{K}
\left[
\frac{1}{N}\sum_{i=1}^N\theta_i(1-d_i)
+
\frac{1}{2N}\sum_{i=1}^N(1-d_i^2)
\right].
\label{eq:qF_start}
\end{equation}
We now simplify the two bracketed terms separately. For the first term,
\begin{align}
\frac{1}{N}\sum_i\theta_i(1-d_i)
&=\bar\theta-\frac{1}{N}\sum_i\theta_i d_i.
\label{eq:theta_term_1}
\end{align}
By the definition of covariance,
\begin{equation}
\gamma
=\frac{1}{N}\sum_i(\theta_i-\bar\theta)(d_i-\bar d)
=\frac{1}{N}\sum_i\theta_i d_i-\bar\theta\bar d,
\label{eq:cov_identity_expand}
\end{equation}
so
\begin{equation}
\frac{1}{N}\sum_i\theta_i d_i
=\bar\theta\bar d+\gamma.
\label{eq:theta_d_identity}
\end{equation}
Substituting this expression into \eqref{eq:theta_term_1},
\begin{equation}
\frac{1}{N}\sum_i\theta_i(1-d_i)
=\bar\theta(1-\bar d)-\gamma.
\label{eq:theta_term_final}
\end{equation}
For the second term,
\begin{align}
\frac{1}{N}\sum_i(1-d_i^2)
&=1-\frac{1}{N}\sum_i d_i^2.
\label{eq:d2_term_1}
\end{align}
Then, by the variance identity,
\begin{equation}
\sigma_d^2
=\frac{1}{N}\sum_i d_i^2-\bar d^2,
\label{eq:variance_identity}
\end{equation}
so
\begin{equation}
\frac{1}{N}\sum_i d_i^2
=\bar d^2+\sigma_d^2,
\label{eq:d2_identity}
\end{equation}
which gives
\begin{equation}
\frac{1}{N}\sum_i(1-d_i^2)
=1-\bar d^2-\sigma_d^2.
\label{eq:d2_term_final}
\end{equation}
Using \eqref{eq:theta_term_final} and \eqref{eq:d2_term_final} in \eqref{eq:qF_start}, it follows that
\begin{equation}
q_F
=\frac{M}{K}
\left[
\bar\theta(1-\bar d)-\gamma
+\frac{1}{2}\{1-\bar d^2-\sigma_d^2\}
\right].
\label{eq:qF_mid}
\end{equation}
The equal-load identity \eqref{eq:bar_d} then implies that
\begin{equation}
1-\bar d=\frac{K}{M},
\label{eq:one_minus_dbar}
\end{equation}
and the baseline-quality term reduces to
\begin{equation}
\frac{M}{K}\bar\theta(1-\bar d)=\bar\theta.
\label{eq:theta_reduces}
\end{equation}
Also,
\begin{equation}
1-\bar d^2
=(1-\bar d)(1+\bar d)
= \frac{K}{M}\left(2-\frac{K}{M}\right).
\label{eq:one_minus_dbar2}
\end{equation}
It follows that
\begin{equation}
\frac{M}{2K}(1-\bar d^2)
=\frac{1}{2}\left(2-\frac{K}{M}\right)
= 1-\frac{K}{2M}.
\label{eq:uniform_fee_base}
\end{equation}
Collecting terms in \eqref{eq:qF_mid} gives
\begin{equation}
q_F
= \bar\theta
+1-\frac{K}{2M}
-\frac{M}{K}\gamma
-\frac{M}{2K}\sigma_d^2
= \boxed{
\bar\theta+1-\frac{K}{2M}
-\frac{M}{2K}(\sigma_d^2+2\gamma)
}.
\label{eq:qF_final}
\end{equation}

\subsection{Deriving the exact quota condition}

Subtract \eqref{eq:qF_final} from \eqref{eq:qQ_detailed}:
\begin{align}
q_Q-q_F
&=
-\frac{K+1}{2(M+1)}
+\frac{K}{2M}
+\frac{M}{2K}(\sigma_d^2+2\gamma).
\label{eq:diff_start}
\end{align}
The first two terms have common denominator $2M(M+1)$:
\begin{equation}
-\frac{K+1}{2(M+1)}+\frac{K}{2M}
=
\frac{-M(K+1)+K(M+1)}{2M(M+1)}
=
\frac{K-M}{2M(M+1)}
=-\frac{M-K}{2M(M+1)}.
\label{eq:flex_gap}
\end{equation}
Therefore,
\begin{equation}
q_Q-q_F
=
-\frac{M-K}{2M(M+1)}
+\frac{M}{2K}(\sigma_d^2+2\gamma).
\label{eq:diff_final}
\end{equation}
The quota yields weakly higher mean reviewed value when $q_Q-q_F\geq0$. Multiplying by the positive quantity $2K/M$ yields
\begin{align}
\sigma_d^2+2\gamma
&\geq
\frac{K(M-K)}{M^2(M+1)} \equiv \tau(M,K).
\label{eq:quota_condition}
\end{align}
Define
\begin{equation}
\tau(M,K)
=\frac{K(M-K)}{M^2(M+1)}.
\label{eq:tau_def}
\end{equation}
Then the quota condition is simply the expression
\begin{equation}
\boxed{\sigma_d^2+2\gamma\geq\tau(M,K).}
\label{eq:quota_condition_box}
\end{equation}
Equality gives identical mean reviewed value; below the threshold the fee yields strictly higher mean reviewed value; above it the quota does.

\subsection{Equivalent variance representation}

Since $c_i=\theta_i+d_i$,
\begin{equation}
\sigma_c^2
=\Var(\theta+d)
=\sigma_\theta^2+\sigma_d^2+2\Cov(\theta,d)
=\sigma_\theta^2+\sigma_d^2+2\gamma.
\label{eq:var_c}
\end{equation}
Consequently,
\begin{equation}
\sigma_c^2-\sigma_\theta^2
=\sigma_d^2+2\gamma.
\label{eq:variance_equiv}
\end{equation}
Substituting into \eqref{eq:quota_condition_box} gives the equivalent condition
\begin{equation}
\boxed{
\sigma_c^2-\sigma_\theta^2\geq\tau(M,K).
}
\label{eq:variance_condition}
\end{equation}
The quota is favoured when the dispersion in overall submission cut-offs is greater than the dispersion in baseline portfolio quality by at least $\tau(M,K)$.

\section{Interpreting Proposition 1}\label{sec:interpret}

\subsection{The fee's flexibility advantage}

If $\sigma_d^2=0$, all authors face the same effective within-author threshold, and Proposition~1 yields
\begin{equation}
q_F-q_Q
=\frac{M-K}{2M(M+1)}>0.
\label{eq:fee_flex_advantage}
\end{equation}
The quota fixes submissions at $K$. Conversely, the fee allows submissions to adjust to the realised quality of each author's candidate set. The threshold $\tau(M,K)$ is the fee-induced distortion needed to offset this flexibility advantage.

\subsection{Why cut-off dispersion hurts the common fee}

With $\bar d$ fixed, lower cut-offs shift more review capacity towards less selective authors. Holding $\gamma$ fixed, under the uniform benchmark, this lowers mean reviewed value by
\[
\frac{M}{2K}\sigma_d^2.
\]
Differences in $d_i$ may reflect fee sensitivity, financing, or cardinal miscalibration.

\subsection{The covariance term}

Let $\gamma=\Cov(\theta,d)$. If $\gamma>0$, stronger authors have higher cut-offs and submit less, lowering $q_F$; if $\gamma<0$, the opposite occurs. If $d_i$ is constant across authors, then $\sigma_d^2=\gamma=0$.

\subsection{A numerical benchmark}

For $M=10$ and $K=2$,
\begin{equation}
\tau(10,2)
=\frac{2(8)}{10^2(11)}
=\frac{16}{1100}
\approx0.014545.
\label{eq:tau_numeric}
\end{equation}
If $\gamma=0$, the quota yields weakly higher mean reviewed value when
\begin{equation}
\sigma_d^2\geq0.014545,
\qquad\text{or equivalently}\qquad
\sigma_d\gtrsim0.1206.
\label{eq:sigmad_numeric}
\end{equation}

\section{Interior identity model and the coefficient-of-variation condition}\label{sec:interior}

Suppose now that
\[
\theta_i=0,
\qquad
g_i(v)=v,
\qquad
\mu_i p<1
\quad\text{for every }i.
\]
Since $V_{ij}=U_{ij}$, the fee condition $U_{ij}\geq\mu_i p$ gives
\begin{equation}
d_i=\mu_i p,
\label{eq:di_mu_p}
\end{equation}
and averaging across authors,
\begin{equation}
\bar d=p\bar\mu,
\label{eq:dbar_p}
\end{equation}
where
\begin{equation}
\bar d=\frac1N\sum_{i=1}^N d_i, \qquad \bar\mu=\frac1N\sum_{i=1}^N\mu_i.
\label{eq:dbar_mubar}
\end{equation}
The equal-load restriction \eqref{eq:bar_d} then gives
\begin{equation}
p\bar\mu=1-\frac{K}{M}=\frac{M-K}{M},
\label{eq:p_mubar}
\end{equation}
from which
\begin{equation}
p=\frac{M-K}{M\bar\mu}.
\label{eq:p_interior}
\end{equation}
Because $d_i=p\mu_i$,
\begin{equation}
\sigma_d^2=p^2\sigma_\mu^2,
\label{eq:sigmad_mu}
\end{equation}
where
\begin{equation}
\sigma_\mu^2=\frac1N\sum_{i=1}^N(\mu_i-\bar\mu)^2.
\label{eq:sigma2bar}
\end{equation}
Also $\theta_i=0$ for all $i$, so $\gamma=0$. Proposition~1 becomes
\begin{equation}
p^2\sigma_\mu^2
\geq
\frac{K(M-K)}{M^2(M+1)}.
\label{eq:cv_step1}
\end{equation}
Substituting \eqref{eq:p_interior},
\begin{equation}
\frac{(M-K)^2}{M^2\bar\mu^2}\sigma_\mu^2
\geq
\frac{K(M-K)}{M^2(M+1)}.
\label{eq:cv_step2}
\end{equation}
Cancelling the positive factor $(M-K)/M^2$,
\begin{equation}
\frac{M-K}{\bar\mu^2}\sigma_\mu^2
\geq
\frac{K}{M+1}.
\label{eq:cv_step3}
\end{equation}
Hence
\begin{equation}
\boxed{
\CV(\mu)^2
\equiv\frac{\sigma_\mu^2}{\bar\mu^2}
\geq
\frac{K}{(M-K)(M+1)}.
}
\label{eq:cv_condition}
\end{equation}
For $M=10$ and $K=2$,
\begin{equation}
\CV(\mu)
\geq
\sqrt{\frac{2}{8\cdot11}}
=\sqrt{\frac{1}{44}}
\approx0.15076,
\label{eq:cv_numeric}
\end{equation}
which corresponds to a coefficient of variation of about $15.1\%$.

\subsection{Interiority condition}

The derivation requires $\mu_i p<1$ for every author. Using \eqref{eq:p_interior},
\begin{equation}
\mu_i\frac{M-K}{M\bar\mu}<1
\qquad\text{for all }i.
\label{eq:interiority_step1}
\end{equation}
The tightest restriction corresponds to the largest $\mu_i$:
\begin{equation}
\boxed{
\frac{\max_i\mu_i}{\bar\mu}
<\frac{M}{M-K}.
}
\label{eq:interiority_condition}
\end{equation}
Note that $\lim_{M\to\infty}M/(M-K) = 1$ for fixed $K$. When the fee sensitivity $\mu_i$ is non-degenerate, it follows that $\max_i\mu_i/\bar\mu>1$, and so the interior condition given above will eventually fail.

\section{Large candidate sets and feasible cut-off dispersion}\label{sec:largeM}

For fixed $K$, the threshold $\tau(M,K)$ converges to zero:
\begin{equation}
\tau(M,K)
=\frac{K(M-K)}{M^2(M+1)}
\sim\frac{K}{M^2}.
\label{eq:tau_asymptotic}
\end{equation}
At the same time, the equal-load restriction forces $\bar d\to1$, and feasible cut-off dispersion shrinks. Since $0\leq d_i\leq1$, it follows that $d_i^2\leq d_i$, and so
\begin{equation}
\sigma_d^2
=\frac{1}{N}\sum_i d_i^2-\bar d^2
\leq\bar d-\bar d^2
=\bar d(1-\bar d).
\label{eq:var_bound1}
\end{equation}
Using $\bar d=1-K/M$,
\begin{equation}
\boxed{
\sigma_d^2
\leq
\left(1-\frac{K}{M}\right)\frac{K}{M}
=\frac{K(M-K)}{M^2}.
}
\label{eq:var_bound2}
\end{equation}
For fixed $K$, $\sigma_d^2\to0$. When baseline-quality dispersion is bounded, the Cauchy-Schwarz inequality gives
\begin{equation}
|\gamma|
=|\Cov(\theta,d)|
\leq\sigma_\theta\sigma_d.
\label{eq:cs_gamma}
\end{equation}
Hence if $\sigma_\theta$ remains bounded and $\sigma_d\to0$, it follows that $\gamma\to0$. If on the other hand $K/M\to\rho\in(0,1)$, then
\[
\bar d\to1-\rho,
\qquad
\tau(M,K)\sim\frac{\rho(1-\rho)}{M}.
\]

\section{Detailed proof of Corollary~1}\label{sec:corollary}

The main paper gives the following result.

\begin{corollaryone}
Let $\theta_i=0$ and $g_i(v)=v$. Suppose a share $s\in(0,1)$ of authors have fee sensitivity $\mu_{\min}$, and let
\[
r=\frac{\mu_{(2)}}{\mu_{\min}}>1,
\]
where $\mu_{(2)}$ is the next-lowest distinct value. Under the equal-load fee,
\[
d_i=\min\{\mu_i p,1\}.
\]
If
\[
M\geq\frac{Kr}{s(r-1)},
\]
all authors with $\mu_i>\mu_{\min}$ have zero submission probability and
\[
q_F=1-\frac{K}{2sM}.
\]
The quota gives higher mean reviewed value if and only if
\[
s(K+1)M<K(M+1),
\]
which holds in particular when $s\leq1/2$.
\end{corollaryone}

\subsection{Cut-offs under the identity model}

With $\theta_i=0$ and $g_i(v)=v$, a candidate manuscript is submitted when
\[
U_{ij}\geq\mu_i p.
\]
Because $U_{ij}\in[0,1]$, the effective cut-off is
\begin{equation}
d_i=\min\{\mu_i p,1\}.
\label{eq:di_boundary}
\end{equation}
A cut-off of one means that the author's submission probability is zero, apart from the zero-probability endpoint event $U=1$. Suppose that only the fraction $s$ of authors with minimum fee sensitivity $\mu_{\min}$ remains active. Let their common interior cut-off be $d$; every other author has cut-off one. Expected total load is then given by
\begin{equation}
sN M(1-d).
\label{eq:active_load}
\end{equation}
Equating this to the quota load $NK$ gives
\begin{equation}
sNM(1-d)=NK,
\qquad
sM(1-d)=K,
\qquad
1-d=\frac{K}{sM},
\end{equation}
so
\begin{equation}
\boxed{d=1-\frac{K}{sM}.}
\label{eq:d_min_group}
\end{equation}
The threshold derived below implies $M>K/s$, so $d\in(0,1)$ as required. Since $d=\mu_{\min}p$,
\begin{equation}
p=\frac{d}{\mu_{\min}}.
\label{eq:p_min_group}
\end{equation}

\subsection{When is the next group priced out?}

Let $\mu_{(2)}$ denote the next-lowest fee sensitivity and define
\[
r=\frac{\mu_{(2)}}{\mu_{\min}}>1.
\]
The next group is priced out when its threshold reaches the upper support point:
\begin{equation}
\mu_{(2)}p\geq1.
\label{eq:next_out_1}
\end{equation}
Using \eqref{eq:p_min_group},
\begin{equation}
\mu_{(2)}\frac{d}{\mu_{\min}} \geq1, \qquad
rd \geq1.
\label{eq:next_out_2}
\end{equation}
Substituting \eqref{eq:d_min_group},
\begin{equation}
r\left(1-\frac{K}{sM}\right)\geq1,
\qquad
r-1 \geq\frac{rK}{sM},
\qquad
M\geq\frac{Kr}{s(r-1)}.
\label{eq:access_threshold}
\end{equation}
Hence
\begin{equation}
\boxed{
M\geq\frac{Kr}{s(r-1)}
}
\label{eq:access_threshold_box}
\end{equation}
is sufficient for the next-lowest group to be priced out; and every author with larger $\mu_i$ is then priced out as well.

\subsection{Mean reviewed value under the fee}

Conditional on submission by an active $\mu_{\min}$ author, $U$ is uniform on $[d,1]$. Hence its conditional mean is
\begin{equation}
\E(U\mid U\geq d)=\frac{1+d}{2}.
\label{eq:truncated_uniform_mean}
\end{equation}
All reviewed manuscripts under the fee come from this group, and so mean reviewed value is
\begin{equation}
q_F
=\frac{1+d}{2}
=\frac{1}{2}\left(2-\frac{K}{sM}\right)
=\boxed{1-\frac{K}{2sM}}.
\label{eq:qF_access}
\end{equation}
Under the quota, Lemma~1 gives
\begin{equation}
q_Q=1-\frac{K+1}{2(M+1)}.
\label{eq:qQ_access}
\end{equation}
The quota yields strictly higher mean reviewed value when
\begin{equation}
1-\frac{K+1}{2(M+1)} > 1-\frac{K}{2sM},
\qquad
\frac{K+1}{M+1} < \frac{K}{sM},
\qquad
s(K+1)M <K(M+1).
\label{eq:access_quality_condition}
\end{equation}
It follows that
\begin{equation}
\boxed{
q_Q>q_F
\quad\Longleftrightarrow\quad
s(K+1)M<K(M+1).
}
\label{eq:access_iff}
\end{equation}

\subsection{A simple sufficient condition}

If $s\leq1/2$, then
\begin{equation}
s(K+1)M
\leq\frac{(K+1)M}{2}.
\label{eq:s_half_1}
\end{equation}
For $K\geq1$ and $M>K$,
\begin{equation}
K(M+1)-\frac{(K+1)M}{2}
 =\frac{2KM+2K-KM-M}{2}
=\frac{(K-1)M+2K}{2}
>0.
\label{eq:s_half_2}
\end{equation}
Therefore
\[
\frac{(K+1)M}{2}<K(M+1),
\]
and $s(K+1)M<K(M+1)$. This proves the stated sufficient condition.

\subsection{Numerical example from the paper}

Take
\[
K=2,
\qquad
r=1.5,
\qquad
s=\frac12.
\]
The exclusion threshold is
\begin{align}
\frac{Kr}{s(r-1)}
&=\frac{2(1.5)}{(1/2)(0.5)}
=\frac{3}{0.25}
=12.
\label{eq:threshold_12}
\end{align}
From $M=12$, only the least fee-sensitive group remains active under the equal-load fee. At $M=12$,
\begin{equation}
d=1-\frac{2}{(1/2)12}=1-\frac13=\frac23.
\label{eq:d_numeric_access}
\end{equation}
The next group has cut-off multiplier $rd=1.5\times(2/3)=1$, exactly attaining the boundary. The mean reviewed value under the fee is
\begin{equation}
q_F=1-\frac{2}{2(1/2)12}=1-\frac{1}{6}=\frac56\approx0.8333,
\label{eq:qf_numeric_access}
\end{equation}
while the quota gives
\begin{equation}
q_Q=1-\frac{3}{2(13)}=1-\frac{3}{26}=\frac{23}{26}\approx0.8846.
\label{eq:qq_numeric_access}
\end{equation}
The difference is about $0.0513$.

\subsection{Asymptotic convergence of mean reviewed values}

From Corollary 1,
\begin{equation}
    q_Q - q_F = \frac{K}{2sM} - \frac{K+1}{2(M+1)}.
\end{equation}
Using the expansion
\begin{equation}
    \frac{1}{M+1} = \frac{1}{M}\left(1+\frac{1}{M}\right)^{-1} = \frac{1}{M}\left(1-\frac{1}{M}+O(M^{-2})\right),
\end{equation}
as $M \to \infty$, it follows that
\begin{equation}
    q_Q - q_F = \frac{K}{2sM} - \frac{K+1}{2}\left[\frac{1}{M}\left(1-\frac{1}{M}+O(M^{-2})\right)\right] = \frac{K-s(K+1)}{2sM}+O(M^{-2}).
\end{equation}
Immediately,
\begin{equation}
    \lim_{M \to \infty} (q_Q - q_F) = 0.
\end{equation}
For $K$ fixed and $s \in (0,1)$, the difference $q_Q - q_F$ can have either sign, although this vanishes as $M \to \infty$. Nevertheless, under the conditions of Corollary 1, the fee continues to exclude all authors outside the least fee-sensitive group.

\section{What is general and what is specific to the uniform benchmark?}\label{sec:generality}

\subsection{Curation and quota quality under a general distribution}

For any realised candidate values, a binding non-transferable quota selects the top $K$ manuscripts whenever perceived returns are strictly increasing in value, and so \eqref{eq:q_monotone_pathwise} is independent of the uniform distribution. Under i.i.d. continuous sampling from any integrable $F$, adding another candidate strictly raises the expected top-$K$ average value.

Suppose now that
\[
V_{ij}=\theta_i+U_{ij},
\]
where $U_{ij}$ are drawn i.i.d. from a distribution $H$, continuous with finite mean. Under the quota,
\begin{equation}
q_Q
=\bar\theta
+\frac{1}{K}\sum_{r=M-K+1}^{M}\E U_{(r:M)},
\label{eq:qQ_generalH}
\end{equation}
and so the quota side continues to be described by upper order statistics.

\subsection{Fee quality under a general within-author distribution}

Let the fee induce an effective threshold $t_i$ in $U$-space. Denote the survival function by
\[
\bar H(t)=1-H(t).
\]
Immediately, the expected load from author $i$ is
\begin{equation}
M\bar H(t_i),
\label{eq:general_load}
\end{equation}
the expected reviewed value is
\begin{equation}
M\left\{
\theta_i\bar H(t_i)
+\int_{t_i}^{\infty}u\,dH(u)
\right\},
\label{eq:general_review_value}
\end{equation}
and the equal-load restriction is
\begin{equation}
\frac{1}{N}\sum_{i=1}^N\bar H(t_i)=\frac{K}{M}.
\label{eq:general_equal_load}
\end{equation}
Under $U[0,1]$, the survival probability is linear in the cut-off and the truncated first moment is quadratic. The fee comparison reduces to the first two cross-author moments of $d_i$ and its covariance with $\theta_i$. For a general distribution $H$, the comparison need not depend only on the first two moments of the cut-offs. The economic mechanism, however, is unchanged.

\section{Scope and boundary cases}\label{sec:extensions}

\subsection{Heterogeneous candidate sets}

If $M_i<K$, author $i$ submits at most $M_i$ manuscripts and leaves some submission rights unused. In that case, the quota does not bind, and the total quota load is then given by
\[
\sum_i \min\{K,M_i\}
\]
rather than $NK$. If $M_i$ (which must be greater than $K$) varies, then each author will still submit their top $K$ candidates. It follows that the gain from curation depends on $M_i$.

\subsection{Unbounded project values}

Corollary~1 relies on bounded values. With values and perceived returns unbounded above, every author retains a positive probability of drawing a candidate that clears any finite fee, so complete exclusion cannot occur.

\subsection{Transferable submission rights}

If submission rights are tradeable, authors with stronger candidate sets can buy more rights. This restores flexibility but no longer reserves review capacity across authors.

\subsection{From reviewed value to welfare}

A full welfare analysis could include the social value of research, author surplus, distributional effects of monetary fees, administrative costs, delay, referee disutility, fee revenue, and the possibility that rejected papers will be resubmitted elsewhere.

  \endgroup
\fi

\end{document}